\documentclass[lettersize,journal]{IEEEtran}
\usepackage{amsthm}
\usepackage{amsmath}
\usepackage{bm}
\usepackage{amssymb}
\usepackage{tabularx}%
\usepackage{ifthen,url}
\usepackage{xcolor}
\usepackage{graphicx,booktabs, colortbl}
\usepackage{subcaption}

\theoremstyle{plain}
\newtheorem{theorem}{Theorem}
\newtheorem{lemma}{Lemma}
\newtheorem{corollary}{Corollary}

\newtheorem{definition}{Definition}

\newcommand*{\etaltwo}{\textit{et al}.\@ }					\newcommand{\ex}[1]{\mathrm{e}{#1}}
\newboolean{SHORT}
\setboolean{SHORT}{TRUE}					
\newcommand{\ifshort}[2]{\ifthenelse{\boolean{SHORT}}{\color{black}#1\color{black}}{\color{black}#2\color{black}}}

\makeatletter
\def\maketag@@@#1{\hbox{\m@th\normalfont\normalsize#1}}
\makeatother

\def\BibTeX{{\rm B\kern-.05em{\sc i\kern-.025em b}\kern-.08em
    T\kern-.1667em\lower.7ex\hbox{E}\kern-.125emX}}
    
\usepackage{balance}

\begin{document}

\title{Dynamics of opinion polarization}
\author{Elisabetta Biondi, Chiara Boldrini, Andrea Passarella, Marco Conti
\thanks{All authors are with the Institute of Informatics and Telematics (IIT) of the National Research Council (CNR), Italy. email: first.last@iit.cnr.it}
\thanks{This work is supported by the European Union – Horizon 2020 Program under the “SoBigData-PlusPlus” (Grant Agreement 871042) and "HumanE-AI-Net" (Grant Agreement 952026) projects. This work is also supported by the SAI project, funded by the CHIST-ERA grant CHIST-ERA-19-XAI-010. The work of C. Boldrini, M. Conti and A. Passarella is partly supported by PNRR - M4C2 - Investimento 1.3, Partenariato Esteso PE00000013 - "FAIR - Future Artificial Intelligence Research" - Spoke 1 "Human-centered AI", funded by the European Commission under the NextGeneration EU programme. The work of C. Boldrini is  also supported by project SERICS (PE00000014) under the MUR National Recovery and Resilience Plan funded by the European Union - NextGenerationEU.

%
}
}


\maketitle

\begin{abstract}
For decades, researchers have been trying to understand how people form their opinions. This quest has become even more pressing with the widespread usage of online social networks and social media, which seem to amplify the already existing phenomenon of polarization. 
In this work, we study the problem of polarization assuming that opinions evolve according to the popular Friedkin-Johnsen~(FJ) model. The FJ model is one of the few existing opinion dynamics models that has been validated on small/medium-sized social groups. First, we carry out a comprehensive survey of the FJ model in the literature (distinguishing its main variants) and of the many polarization metrics available, deriving an invariant relation among them. Secondly, we derive the conditions under which the FJ variants are able to induce opinion polarization in a social network, as a function of the social ties between the nodes and their individual susceptibility to the opinion of others. Thirdly, we discuss a methodology for finding concrete opinion vectors that are able to bring the network to a polarized state. Finally, our analytical results are applied to two real social network graphs, showing how our theoretical findings can be used to identify polarizing conditions under various configurations. 
\end{abstract}


\section{Introduction}
\IEEEPARstart{W}{ith}  the rise of social media and online social networks, online interactions have started playing an increasingly important role in how people form their opinions, to the point that news consumption itself is now often mediated by social interactions~\cite{hermida2012share,gottfried2016news}. Social networks, though, do not merely provide a transparent technological substrate that facilitates interactions in the online dimension. Their algorithmic personalization, aimed at highlighting content that is more interesting to each of us, effectively reinforces our cognitive biases, reducing the cognitive discomfort we experience when exposed to opinions challenging our beliefs but at the same time reducing the diversity and range of opinions we are exposed to. By reinforcing consonant opinions and downplaying, or even removing, discordant ones, social networks cradle us into curated filter bubbles and comfortable echo chambers. However, whether this leads to actual polarization~\cite{spohr2017fake,fiorina2008political,conover2011political} is still debated. Some argue that the very nature of social networks, i.e., the socialization of information consumption, may counteract the above effects~\cite{Messing2014}, others that individual choices (to bond with similar others and to prefer concordant information) are more predominant than algorithmic filtering~\cite{Bakshy2015}, others again that exposure to opposing views is more likely to actually backfire than to widen our perspectives~\cite{Bail2018}. To make matter worse, information may not only be partisan but it could also be blatantly fake~\cite{Lazer2018}. 

This quest towards a better understanding of the impact of the social algorithm~\cite{Lazer2018} and misinformation on our societies is ingrained with a more general question that, even when removing the cyber-dimension, still remains unsolved: how do people form their opinions? 
This question has fascinated sociologists and economists alike since much before the advent of the Internet, but it has recently gained new momentum, with computational sociologists and control theorists now weighing in. 
The literature on opinion dynamics is vast, with many models being proposed that aim at capturing a variety of cognitive and social mechanisms that lead to forming an opinion, such as social influence (which determines whose opinion you are affected by), cognitive dissonance (which triggers your willingness to adapt), anchoring to one’s own opinion (which captures our prejudices). For an in-depth discussion, we refer the interested reader to recent surveys, such as~\cite{DONG201857,Anderson2019,HASSANI202222}.

So-called \emph{averaging models} are one of the most popular classes of such opinion dynamics models~\cite{degroot1974reaching,french1956formal,friedkin1990social}. 
In these models, the final opinions (also known as expressed opinions) are a function of a repeated weighted averaging of the opinions of neighboring (in the influence graph) nodes. 
The strengths of averaging models lie in their mathematical tractability~\cite{proskurnikov2017tutorial}, ability to capture strong\footnote{This is in contrast with the weak opinion diversity generated by models like the Hegselmann-Krause model~\cite{Hegselmann2002}, where the final opinions form clusters in which every opinion is the same.} opinion diversity~\cite{Mas2014}, and their general flexibility (e.g., they can capture the \emph{wisdom of the crowd} phenomenon~\cite{Das2013} or include prominent agents~\cite{Acemoglu2011} such as media sources and politicians that may be systematically biased and not willing to change their opinion at all).

The Friedkin-Johnsen (FJ) model~\cite{friedkin1990social} is the most popular averaging model in the related literature. It is the only model that has been validated on small and medium-sized groups~\cite{Friedkin2011,Friedkin2017}, and even in human-AI group experiments~\cite{Askarisichani2022}. Focusing on it, our first contribution is to provide a comprehensive review of all the major variants of the FJ model and of the polarization metrics described in the related literature. For them, we will highlight their key features and the differences between each other. We found that polarization metrics are linked together through an invariant relationship.
As a second contribution, we derive the conditions under which the FJ model yields polarization, for each of the polarization metrics identified before. In addition, we also prove that the polarizing opinion vectors can be found analytically in most cases. All the results obtained are exploited to identify polarizing conditions, under different configurations, with two popular datasets of real social networks. 

\vspace{-10pt}
\subsection{Background and motivation}

The simplest averaging model is the DeGroot's model~\cite{degroot1974reaching}, whereby the opinion of a node is simply the average opinion of its neighbours, weighted by the strength of their social influence. 
This model, however, is not considered realistic, since, when it converges (i.e., if the nodes' opinions stabilize), it always leads to \emph{consensus}, i.e., to a final state in which all nodes have exactly the same opinion~\cite{krackhardt2009plunge}. 
To overcome this problem, Friedkin and Johnsen~\cite{friedkin1990social} proposed a variation on the Groot's model that introduces a certain degree of stubbornness in nodes.
%
Their hypothesis is that a personal opinion always remains at least partly anchored to the initial opinion (or \emph{prejudice}), more or less so depending on the individual's attitude to be influenced by others. The Friedkin-Johnsen model does not lead to consensus (except in very particular cases~\cite{Friedkin2015}) and 
%
has been widely popular in the related literature~\cite{Gionis2013,Bindel2015,Matakos2018,Musco2018,Chen2018,proskurnikov2017tutorial,proskurnikov2018dynamics,Friedkin2017,Friedkin2015}.
The FJ model has enjoyed two main avenues of research: on the one hand, the derivation of the conditions for convergence or consensus has been the main focus of the research efforts from the control theory domain~\cite{parsegov2016novel,proskurnikov2018dynamics,proskurnikov2017tutorial,Proskurnikov2018}. 
On the other hand, the graph-theoretical efforts~\cite{Bindel2015,Matakos2017,Chen2018,Dandekar2013,Gionis2013, Musco2018,Abebe2018} have been focused on understanding the effects of the underlying influence graph on opinion formation, polarization, and on how to interfere with the opinion formation process in order to obtain a desired outcome, (e.g., shifting the opinion in a specific direction, minimizing polarization and/or disagreement). 

While all the above works refer to the opinion dynamics model they leverage as Friedkin-Johnsen, they are often relying on a simplified version of it. Specifically, they use the more mathematically tractable version (which we refer to, later on, as rFJ), which, however, is not able to capture polarization (we discuss this point later in the paper).   
This has resulted in great confusion regarding which finding holds under which hypothesis. 
The second gap in the related literature, and a direct consequence of the above confusion, lies in whether the FJ model is actually able to capture polarization or not. Indeed, despite being opinion polarization a fundamental feature of a realistic opinion formation process, only Gionis~\etaltwo~\cite{Gionis2013} and Dandekar~\etaltwo~\cite{Dandekar2013} have explicitly tackled this problem. Analyzing the problem on undirected social networks, they have proved that two \emph{variants} of the FJ model are neither capable of changing the average opinion of the social network nor of increasing the weighted difference of opinions among nodes of the same neighbourhood. However, what happens with the general FJ model and with other polarization metrics is yet unknown.

\section{Modelling framework}\label{sec:model}

We explicitly differentiate between the social graph and the influence graph. They both comprise the same set $\mathcal{V}$ of $n$ vertices and the same set of edges $\mathcal{E}$, but the weights of the edges are different and have different meanings. The social graph, denoted with $\mathcal{S}$, represents people (vertices) and the \emph{social} relationships between them (through the edge weights $\hat{w}_{ij}$).
The strength of a social relationship is typically measured in terms of the number of interactions that two people have~\cite{zachary1977information} and for this reason the few results on polarization in the related works assume that the social graph is undirected~\cite{Dandekar2013,Gionis2013}. 
In this paper, we will consider the general case of a directed social graph, specifying how results change in the specific case of an undirected one.
The \emph{influence graph} $\mathcal{I}$ describes how a node's opinion is influenced by that of its neighbours. The existence of an edge from node $i$ to node $j$ in $\mathcal{I}$ implies that node $j$ exerts an influence on the opinion of node $i$, and the strength of this influence is expressed by the edge weight $w_{ij}$. Lacking additional information, the influence graph can be derived from the social graph, leveraging the intuition that stronger social relationships will influence more than weak ones. Specifically, starting from the social weights $\hat{w}_{ij}$, the influence $w_{ij}$ can be computed as $w_{ij}~=~\frac{\hat{w}_{ij}}{\sum_{j=1}^{n} \hat{w}_{ij}}$. Please note that this definition is the only one that allows a unique correspondence between all the variants of the FJ model.
The matrix  $W = (w_{ij})$ is called  \emph{influence matrix} and is assumed to be row-stochastic (because it captures how the influence a node is subject to is split among its neighbours). 
The influence matrix is in general asymmetric (corresponding to a directed influence graph), even starting from a symmetric social matrix $\hat{W} = (\hat{w}_{ij})$, because the influence weight $w_{ij}$ expresses the relative importance of $j$ with respect to all $i$'s social relationships. Hence, the same social relationship intensity can weigh very differently depending on the strength of other relationships.

\vspace{-10pt}
\subsection{The Friedkin-Johnsen family of opinion dynamics models}
\label{sec:model:opinion_dynamics}

A discrete-time opinion dynamics model tracks the evolution of $z_i(k)$, the opinion expressed by a node $i$ at time $k$. Opinions are generally assumed to be real-valued, i.e., continuous in a certain reference interval. Similarly to the related literature~\cite{Friedkin2011,Gionis2013,Bindel2015}, here we assume that opinions belong to~$\left[-1,1\right]$. Thus, extremes -1 and 1 represent opposing viewpoints on an issue. For a given configuration of its input parameters, the model is said to be \emph{convergent} if $z_i(k+1) \rightarrow z_i$ for all $i$ as $k$ grows to infinity. A convergent model is said to \emph{reach consensus} if $z_i(k+1) \to z$ for all $i$ as $k$ grows to infinity.
In the Friedkin-Johnsen model family, before the opinion formation process starts, each node $i$ has an initial opinion $s_i$, often referred to as \emph{internal} or \emph{fixed} opinion (or \emph{prejudice}). In contrast, the opinion $z_i(k)$ is often referred to as the \emph{expressed opinion} at time $k$. \ifshort{}{We start by defining the more general version of the model (Definition~\ref{def:gfj} below), which we denote with the acronym gFJ.}
In Table~\ref{tab:fj_family} we summarize the variants of the FJ models that can be found in the literature and we discuss them separately hereafter. We denote with $N(i)$ the neighborhood set of node $i$.

\begin{table}[h]
    \centering
    \caption{The FJ family of models}
    \label{tab:fj_family}
    \rule{\columnwidth}{\heavyrulewidth}
    \begin{flalign}
    &\text{gFJ:} & z_i(k + 1) = (1-\lambda_{i}) s_i + \lambda_i \sum_{j \in \{i\} \cup N(i)} w_{ij} z_{j}(k) &&& \label{eq:gfj}
    \end{flalign}
\vspace{-\baselineskip}
    \begin{flalign}
    &\text{vFJ:} & z_i(k + 1) = \frac{\hat{w}_{ii} s_i + \sum_{j \in N(i)} \hat{w}_{ij} z_{j}(k)}{ \hat{w}_{ii}+\sum_{j \in N(i)} \hat{w}_{ij}} &&& \label{eq:vfj}
    \end{flalign}
\vspace{-\baselineskip}
    \begin{flalign}
    &\text{rFJ:} & z_i(k + 1) = \frac{s_i + \sum_{j \in N(i)} \hat{w}_{ij} z_{j}(k)}{1 + \sum_{j \in N(i)} \hat{w}_{ij}}  &&& \label{eq:rfj}
    \end{flalign}
\vspace{-\baselineskip}
\rule{\columnwidth}{\heavyrulewidth}
\vspace{-\baselineskip}
\end{table}

\subsubsection{The generalized Friedkin-Johnsen model - gFJ}
\eqref{eq:gfj} in Table~\ref{tab:fj_family} corresponds to the more general version of the model originally proposed by Friedkin and Johnsen~\cite{friedkin1990social}. The outermost weighted average depends on parameter $\lambda_i$, corresponding to the \emph{susceptibility} of node $i$ to the opinions of other nodes. The innermost weighted average depends on the \emph{influence} $w_{ij}$ that node $j$ exerts on node $i$.
Two main mechanisms are at play here: \emph{anchoring}, to the node $i$'s internal opinion $s_i$,  and variable \emph{susceptibility} $\lambda_i$, to other nodes' opinions. Nodes with zero susceptibility value are  \emph{stubborn} nodes and they never change their opinion. 
A common matrix-formulation of the model is the following:
\begin{equation}\label{eq:gfj_matrix}
\bm{z}(k + 1) = (I - \Lambda) \bm{s} + \Lambda W \bm{z}(k),
\end{equation}
where $\Lambda$ is a diagonal matrix containing the susceptibility values $\lambda_{i}$, while $W$ is the influence matrix. 
Note that the opinion of a node $i$ depends both on its initial prejudice $s_i$ (by a weight $1-\lambda_i$) and on its current opinion (by a weight $\lambda_i w_{ii}$). The only case in which this does not happen is when the node is stubborn ($\lambda_i = 0$) or when $w_{ii}=0$. $\Lambda$ and $W$ are sometimes coupled via the condition $1 - \lambda_i = w_{ii}$~\cite{proskurnikov2017tutorial}, however we do not make this assumption here. 
%
The conditions under which the gFJ model achieves convergence and consensus have been thoroughly studied in the related literature~\cite{Friedkin2015,parsegov2016novel,proskurnikov2017tutorial}. A sufficient condition for convergence~\cite{proskurnikov2018dynamics} is reported below, which we will use often in the rest of the paper.
\begin{theorem}[Sufficient condition for the gFJ]\label{theo:gfj_convergence}
If $\Lambda W$ is stable (i.e., has eigenvalues inside the open unit circle $\{z\in \mathbf{C}: |z|<1\}$), the gFJ model is convergent and its only stationary point $\bm{z}$ (i.e., steady-state solution) is given by the following:
\begin{equation}\label{eq:gfj_solution}
  \bm{z}= (I-\Lambda W)^{-1} (I-\Lambda) \bm{s}.
\end{equation}
%
\end{theorem}

\noindent
We refer the reader to the SI Appendix for a brief summary of the main findings on the topic of opinion convergence.
\subsubsection{The variational Friedkin-Johnsen model - vFJ}
Dandekar~\etaltwo~\cite{Dandekar2013} and  Matakos~\etaltwo~\cite{Matakos2017} use a variant of FJ that we call the variational Friedkin-Johnsen model (vFJ), whose update function can be found in \eqref{eq:vfj} of Table~\ref{tab:fj_family}. According to this model, the current opinion of a node is the weighted average between its prejudice and the current opinion of the other nodes. Thus, in this variant of the FJ model, the current opinion of the node itself is not taken into account. We can formulate the expressed opinion in matrix form in the following way:
\begin{equation}
     \bm{z}=(D+\tilde{A}-A)^{-1}\tilde{A}\bm{s},
     \label{eq:sol_vFJ}
\end{equation}
where $D$ is the diagonal degree matrix ($\sum_j \hat{w}_{ij}$ for the $i$-th diagonal element), $A$ is the adjacency matrix (whose $i,j$ element is $\hat{w}_{ij}$ and the diagonal is null) and $\tilde{A}$ is a diagonal matrix whose $i$-th diagonal entry is equal to $\hat{w}_{ii}$. To model stubborn nodes, we can admit $\hat{w}_{ii}$ to be equal to $\infty$. In this case, matrix $\tilde{A}$ contains infinite values and~\eqref{eq:sol_vFJ} should be treated as discussed in the SI Appendix. 
The relation between vFJ and gFJ has never been explicitly discussed in the related literature, where the two are implicitly treated as interchangeable and generically referred to as Friedkin-Johnsen model. However, the two models are not mathematically equivalent: the vFJ does not include node $i$'s current opinion~$z_i$ in the averaging process, while gFJ pools both the initial opinion $s_i$ and the current opinion $z_i$\footnote{Note that the coupling condition $\lambda_i = 1 - w_{ii}$ makes no sense for vFJ since the weight of node $i$'s current opinion is zero, so there is nothing to couple.}. 
The different flexibility of the two models becomes clear when observing that while the vFJ only features the matrix $\hat{W}= (\hat{w}_{ij})$ as parameters of the model (leading to a maximum $n^2$ degrees of freedom, with $n=|\mathcal{V}|$), the gFJ includes also matrix $\Lambda$, thus in total its degrees of freedom are $n^2+n$.  From a practical point of view, however, the only difference between the two models is the parameter $w_{ii}$, which, in gFJ, takes into account  node $i$'s opinion~$z_i$ in the averaging process, as we will see in the proof of Corollary~\ref{theo:vFJ}. 

\subsubsection{The restricted Friedkin-Johnsen model - rFJ}

The vFJ model with $\hat{w}_{ii}$ set to 1 as in~\eqref{eq:rfj} of Table~\ref{tab:fj_family} is very popular in the related literature, mainly due to its mathematical tractability. The model has been used in~\cite{Bindel2015,Gionis2013,Musco2018,Chen2018,Chitra2019}. 
The main difference between the rFJ and the vFJ model is the absence of the weight for $s_i$, so the parameters are only $\hat{w}_{ij}$ for all $i\neq j$ thus implying $n^2-n$ degrees of freedom. Note that, since the weights $\hat{w}_{ij}$ are free to vary ($\hat{w}_{ij} \ge 0$), it is impossible to control the susceptibility (i.e., the importance of one's own initial opinion), even indirectly. 
A common matrix-formulation of the rFJ model is the following:
\begin{equation}\label{eq:rfj_matrix}
(D+I) \bm{z}(k + 1) = \bm{s} + A \bm{z}(k),
\end{equation}
where $D$ and $A$  are defined as described for vFJ. The solution to the above problem can be written as $\bm{z} = (L+I)^{-1} \bm{s}$, where $L = D - A$ is the Laplacian matrix.
The formulation of the rFJ model is particularly convenient from a mathematical standpoint (since $L+I$ is symmetric and many useful matrix formulas leverage symmetry), and this is the reason why it has been so often used in the related literature. 

\subsubsection{The matrix representation of the FJ model}

In the whole set of FJ models, the final opinion $\bm{z}$ of the opinion formation process can be expressed as $\bm{z}=H\bm{s}$, where $H$ is a matrix that varies depending on the specific FJ version considered, whose formulas are summarized in Table~\ref{tab:h_fj}. In the remaining of the paper, we will see that those matrices will be the key to the analysis of FJ polarization.

\begin{table}[h]
    \centering
    \caption{Matrix $H$ for the different FJ models}
    \label{tab:h_fj}
    \rule{\columnwidth}{\heavyrulewidth}
    \begin{flalign}
    &\text{gFJ:} & H_{g}= (I-\Lambda W)^{-1} (I-\Lambda) &&& \label{eq:h_gfj}\\
    &\text{vFJ:} & H_{v}=(D+\tilde{A}-A)^{-1}\tilde{A}  &&& \label{eq:h_vfj}\\
    &\text{rFJ:} & H_{r}=(D+I-A)^{-1}  &&& \label{eq:h_rfj}
    \end{flalign}
    \vspace{-\baselineskip}
    \rule{\columnwidth}{\heavyrulewidth}
\vspace{-\baselineskip}
\vspace{-10pt}
\end{table}
\subsection{Polarization metrics}
\label{sec:polarization}
In an opinion formation process, polarization is observed when there is a variation in a target index $\Phi$ (any, e.g., of the indices in Definition~\ref{def:polarization indices}) between the initial opinion and the final opinion of nodes. A rigorous definition is provided in the following. Note that Definition~\ref{def:polarization_def} below is basically an abstraction of the polarization definitions in the related literature. In fact,  while related works typically focus on a specific polarization metric and define polarization based on it, here we abstract the metric into the variable $\Phi$ and we provide a general definition that holds for all the polarization metrics discussed later on in Definition~\ref{def:polarization indices}.
\vspace{-5pt}
\begin{definition}[Polarization]\label{def:polarization_def} For a polarization index~$\Phi$, we say that the opinion formation model $\mathcal{M}$ is \emph{$\Phi$-polarizing } or \emph{polarizing for $\Phi$} if it exists at least an initial opinion vector $\bm{s}$ such that the corresponding final opinion vector $\bm{z}$ satisfies the following inequality: 
\begin{equation}
\Phi(\bm{z}) > \Phi(\bm{s}). 
\end{equation}
In this case, we say that $s$ yields to $\Phi$-polarization and we call it \emph{polarizing vector} or \emph{polarizing prejudice}; its induced polarization is measured in terms of the polarization shift, i.e by the function $\Delta_\Phi$ defined as:
\begin{equation}
    \Delta_\Phi(\bm{s})=\Phi(\bm{z})-\Phi(\bm{s}),
    \label{eq:delta}
\end{equation}
If the model $\mathcal{M}$ is not polarizing, we say that it is \emph{$\Phi$-depolarizing } or \emph{depolarizing for $\Phi$}. 
\end{definition}
\noindent Please observe that the definitions of polarizing and depolarizing model $\mathcal{M}$ are not symmetric: to depolarize, a model $\mathcal{M}$ should let opinions evolve in such a way that, at the end of the process, $\Phi$ is always decreasing for all possible choices of internal opinions $(s_i)$; instead,  $\mathcal{M}$ is polarizing if $\Phi$ does not decrease for at least one initial opinion vector $\bm{s}$. The justification of the asymmetry lies in the importance of determining whether a model \emph{can} capture the polarization phenomenon, which means that it does it in at least one case. Please note that, for the sake of brevity, in the following we may simply refer to the opinion vector as \emph{opinion}, omitting the word ``vector". 

For the polarization index $\Phi$, the related literature has explored several different metrics, each capturing a different property of an opinion vector. 
Below we have collected the most popular definitions, for which we provide a short discussion. 

\begin{definition}\label{def:polarization indices} For an opinion $\bm{x}=(x_i)\in [-1,1]^n$ the following polarization indices are defined:
\begin{align}
NDI(\bm{x}) &= \sum_{(i,j) \in \mathcal{E}} w_{ij} (x_i - x_j)^2 \label{eq:ndi}\\
GDI(\bm{x}) &= \sum_{i,j \in \mathcal{V} : i<j}  (x_i - x_j)^2 \label{eq:gdi}\\
P_{1}(\bm{x}) &= \sum_{i \in \mathcal{V}} (x_i - \bar{\bm{x}})^2 = \| \bm{x} - \bar{\bm{x}}\|_2^2 \label{eq:p1}\\
P_{2}(\bm{x}) &= \frac{1}{|\mathcal{V}|} \sum_{i \in \mathcal{V}} x_i ^2 = \frac{1}{|\mathcal{V}|} \| \bm{x} \|_2^2 \label{eq:p2}\\
P_{3}(\bm{x}) &=  \sum_{i \in \mathcal{V}} x_i ^2 = \| \bm{x} \|_2^2 \label{eq:p3}\\
P_4(\bm{x}) &=  \sum_{i \in \mathcal{V}} |x_i| = \| \bm{x} \|_1 \label{eq:p4}
\end{align}

%
\end{definition}

\noindent
The Network-disagreement Index (NDI)~\cite{Dandekar2013, Musco2018,Bindel2015,Chen2018,Matakos2017} is the sum, over all  nodes, of the weighted disagreement in each node pair, which represents (except for the division by $n$) the average disagreement in the network as a whole.  NDI is the only topology-dependent metric,  in the sense that the same opinions may give rise to a completely different NDI depending on how the vertices are connected.
The Global Disagreement Index~\cite{Dandekar2013} (GDI) measures the conflict between all the users in the network, regardless of whether they share a social link or not. 
$P_1$~\cite{Musco2018}, corresponding to the mean-centered 2-norm of opinions, measures the polarization as a deviation of the opinions from the average. 
The definitions of $P_2$~\cite{Matakos2017} and $P_3$~\cite{Chen2018}, instead, intend the polarization as the deviation from the complete neutrality, represented with the value 0 (the middle ground between the two extremes -1 and 1).  
Finally, $P_4$ is referred to as \emph{total absolute opinion} and has been introduced by Friedkin and Johnsen~\cite{Friedkin2011}. While all previous indices were related to 2-norms, the total opinion is equivalent to the 1-norm. Similarly to $P_2$ and $P_3$, the index $P_4$ measures the ``absolute total'' opinion in the network and has the same semantic: it measures the deviation from the neutrality (represented by 0). 
While not directly a measure of polarization, the concept of choice shift caused by the opinion formation process (see definition below) is sometimes used in the related literature as an intermediate step in gauging the direction towards which opinion moves. 


\begin{definition}[Choice shift]\label{def:choice_shift}
A choice shift occurs when the mean attitude of the group at the end is different from the mean attitude at the beginning:
\begin{equation}
    \sum_i z_i \neq \sum_i s_i.
\end{equation}
\end{definition}

\noindent
The choice shift has been analyzed, for rFJ, by Gionis \etaltwo in~\cite{Gionis2013}, where it is found that, if the social graph is undirected ($w_{ij} = w_{ji}$), changing the graph topology will not determine a choice shift. In the following, we will discuss if this finding carries over to gFJ and under which conditions.

\subsubsection{Polarization invariants}

The above polarization indices have been introduced in the literature mostly as standalone metrics. In the remaining of the section, we establish equivalence relationships among them (Lemmas~\ref{lemma:invariant_p1gdi}-\ref{lemma:invariant_p2p3}) and we derive a polarization invariant (Lemma~\ref{lemma:invariant_all}). 

\begin{lemma}\label{lemma:invariant_p1gdi}
It holds that $GDI(\bm{x})= |\mathcal{V}| \cdot P_1(\bm{x})$, thus the two metrics $GDI$ and $P_1$ are equivalent. 
\end{lemma}

\begin{proof}
See SI Appendix.
\end{proof}

\begin{lemma}\label{lemma:invariant_p2p3}
It holds that $P_3(\bm{x})= |\mathcal{V}|  \cdot P_2(\bm{x})$, thus the two metrics $P_2$ and $P_3$ are equivalent. 
\end{lemma}

\begin{proof}
Differently from Lemma~\ref{lemma:invariant_p1gdi}, the proof is trivial and the thesis can be derived straightforwardly from Definition~\ref{def:polarization indices}.
\end{proof}

\noindent
Leveraging the results above, we can classify the polarization indices into four main classes of equivalence (Table~\ref{tab:classes_polarization}), in the sense that the behavior of a model is invariant in each class.

\begin{table}[h!]
\centering
\caption{Classes of polarization}
\begin{tabular}{lcc}
\toprule
Type & What is captured & Indices\\
\midrule
\textbf{Local} & opinion spread among neighboring nodes & $NDI$\\
\textbf{Dispersion} & opinion spread among all nodes & $GDI$, $P_1$\\
\textbf{Absolute} & quadratic closeness to the extremes & $P_2$, $P_3$\\
\textbf{Total} &  linear closeness  to extremes & $P_4$\\
\bottomrule
\end{tabular}
\label{tab:classes_polarization}
\end{table}
The four classes capture four different concepts of polarization. However, they are correlated by the following important invariant that will be used in the next section and whose proof is given in the SI Appendix. 

\begin{lemma}[Polarization invariant]\label{lemma:invariant_all}
For all opinion vectors~$\bm{x}$, the following inequality holds:
\begin{equation}
    P_1(\bm{x})\geq P_3(\bm{x}) - \frac{P_4(\bm{x})^2}{|\mathcal{V}|}.
    \label{eq:invariant_pol_neg}
\end{equation}
\end{lemma}

\noindent
From this relation, the following corollary follow, whose proof is provided in the SI Appendix.
\begin{corollary}\label{coro:invariant_nochoiceshift}
When there is no choice shift the polarization moves in the same direction for both $P_1$ and $P_3$. 
\end{corollary}
Corollary~\ref{coro:invariant_nochoiceshift} says that, when there is no choice shift, the polarization with $P_1$ implies the polarization with $P_3$ and vice versa. 
A straightforward remark is that, in all the cases when the choice shift is null (for example when opinions are positive and $W$ is symmetric, as shown by Gionis~\etaltwo~\cite{Gionis2013}), the Dispersion and Absolute classes of polarization are identical and represent the only global class of polarization.  

\section{gFJ is globally polarizing but locally depolarizing}
\label{sec:gFJ_all}

We start by focusing on the most general Friedkin-Johnsen model, the gFJ, and we investigate whether in this case the dynamics of the process lead to polarization or not. 

\vspace{-10pt}

\subsection{Polarization under NDI}
The first result is about the local polarization captured by the NDI index.



\begin{theorem}[gFJ: local polarization with NDI]\label{theo:gFJ_ndi} 
The gFJ model is always depolarizing with respect to $NDI$, in the sense that, for every prejudice $\bm{s}$, we have that $NDI(\bm{z})\leq NDI(\bm{s})$.
\end{theorem}

\begin{proof}
As stated in Theorem~\ref{theo:gfj_convergence}, the gFJ model converges to the vector $\bm{z}$ obtained from $\bm{z}=\left(I-\Lambda W\right)^{-1}(I-\Lambda)\bm{s}$.
For each node $i$, consider the following cost function:
\begin{equation}
f_i(z_i)= (1-\lambda_i)(s_i-z_i)^2+ \lambda_i\sum_{j=1}^n  w_{ij} (z_i-z_j)^2,\label{eq:cost}
\end{equation}  
which penalizes opinion $z_i$ if far from $s_i$ ($i$'s initial prejudice) and from $\sum_{j=1}^n w_{ij} z_j$ (the mean opinion of $i$'s neighborhood). We can prove that the expressed opinion $(z_i)_i$ of gFJ provided by~\eqref{eq:gfj_solution} is the Nash Equilibrium of cost function~\eqref{eq:cost} (for details, please refer to the SI), i.e. $z_i$ minimizes $f_i$ for all~$i$, so that $f_i(z_i)\leq f_i(s_i)$ for all $i$. Since $NDI(\bm{z})\lambda_i\leq\sum_i f_i(z_i)$ and $NDI(\bm{s})\lambda_i=\sum_i f_i(s_i)$, we obtain that $NDI(\bm{z})\leq NDI(\bm{s})$ and, as a consequence, it follows that the gFJ is $NDI$-depolarizing.
\end{proof}

The result described above is intuitive: by  definition, gFJ captures the willingness of each node to reduce the conflict (weighted by the matrix $W$) caused by the discordance of opinions with its neighbours, which is exactly what $NDI$ measures. For this reason, the gFJ model is depolarizing in a local sense, but this however does not imply anything about global polarization. On the contrary, we will prove that gFJ can be polarizing at the global level depending on the interplay between the social network weights and nodes' susceptibility to the opinion of others. This is a key result, since it proves that gFJ does capture the polarization phenomenon in social networks. 

\vspace{-10pt}

\subsection{Polarization under $P_2$, $P_3$, and $P_4$}
We start by deriving the conditions under which gFJ is polarizing for the global metrics $P_2$, $P_3$, and $P_4$ (the proof is provided in the SI Appendix).  




\begin{theorem}[gFJ: global polarization with $P_2, P_3, P_4$]\label{theo:gFJ_p2p3p4}
gFJ is \emph{polarizing} with $P_2, P_3, P_4$ if and only if matrix $H_g$ defined  in~\eqref{eq:h_gfj} is not doubly stochastic (i.e., a square nonnegative matrix, each of whose rows and columns sums to 1). 
Furthermore, we can distinguish the following two cases:
\renewcommand{\labelenumi}{(\roman{enumi})}
\begin{enumerate}
    \item if there are naive nodes (i.e., $\exists i \in \mathcal{V} : \lambda_i = 1$), matrix $H_g$ is never doubly stochastic and thus gFJ is polarizing;
    \item if there are no naive nodes (i.e., $\forall i \in \mathcal{V}, \lambda_i < 1$), matrix~$H_g$ is not doubly stochastic, and equivalently gFJ is polarizing with $P_2, P_3, P_4$, if and only if the following condition holds true for at least one node $i \in \mathcal{V}$:
\begin{equation} \label{eq:stochasticity_cond}
    \sum_{j \in \mathcal{V}}\frac{\lambda_j w_{ji}}{1-\lambda_j}\neq \frac{\lambda_i}{1-\lambda_i}.
\end{equation}
\end{enumerate}
\end{theorem}

Intuitively, the fact that $H_g$ is not double-stochastic is a measure of the presence of nodes that are more influential than others. 
This is straightforward to see in the case of naive nodes (Theorem~\ref{theo:gFJ_p2p3p4}.(i)), where all the non-naive nodes play the role of influencers (because they are always able to sway the naive nodes' opinions towards theirs), potentially increasing the polarization. 
%
When there are no naive nodes, the intuition behind  Theorem~\ref{theo:gFJ_p2p3p4} is more difficult to grasp. Let us split the effect of social influence and individual susceptibility. To isolate the former, let all nodes have the same susceptibility~$\lambda$. Since \eqref{eq:stochasticity_cond} is reduced to $\sum_{j \in \mathcal{V}} w_{ji} \neq 1$, $W$ not being double-stochastic becomes the condition for polarization, which corresponds to the case where the social influence out of any node $i$ is equivalent to the incoming social influence. However, in the general case, pure social influence is dampened by individual susceptibility: stubborn nodes are not swayed, regardless of the social influence they are subject to. The condition in \eqref{eq:stochasticity_cond} exactly captures this interplay between susceptibility and social influence. 

When gFJ is depolarizing, it is also unable to produce choice shift, as the following corollary states. 
\begin{corollary}\label{cor:H_doubly_stoch}
When $H_g$ is doubly stochastic and thus gFJ is depolarizing with $P_2$, $P_3$, $P_4$, for all opinion vector $\bm{s}$, it holds that $ P_4(\bm{z})=P_4(\bm{s})$ and $\sum_i z_i = \sum_i s_i$.
\end{corollary}

\subsubsection{How to find polarizing opinion vectors}
In Theorem~\ref{theo:gFJ_p2p3p4} we have derived the sufficient and necessary condition for gFJ to be polarizing. We can give a first characterization of the polarizing vectors with $P_2, P_3$, and  $P_4$: they can always be chosen with concordant entries (i.e., $\textrm{sgn}(s_i) = \textrm{sgn}(s_j), \forall i,j$). 
\noindent For a vector $\bm{x}=(x_i)_i$, we will indicate with $\bm{x}^\text{abs}$ the vector with positive entries given by $\bm{x}^\text{abs} = \left(\left|x_i\right|\right)_i $.
This is what is affirmed by the following theorem.  
\begin{theorem}[$P_2,P_3,P_4$ polarizing vectors can be obtained with concordant entries]\label{theo:polarization_positive_entries} 
Whenever the model is polarizing with $P_i$,  $i=2,3,4$, and $\bm{s}$ is a polarizing opinion vector, $\pm\bm{s}^\textrm{abs}$ (that have concordant entries) are polarizing vectors inducing greater or equal (than that of s) polarization.
Furthermore,  if the network has no naive nodes and $\Lambda W$ is \emph{irreducible} (i.e. the graph induced by $\Lambda W$ is strongly connected), the polarizing opinion vector that \emph{maximises} the polarization has concordant entries.
\end{theorem} 
\noindent
This result is important and also pretty intuitive: since the $P_{2,3,4}$ polarization captures the shift from a neutral state (close to $0$) to an extreme state (close to $1$ or $-1$), the polarization calculated on the same vector with all the entries concordant must be greater or equal, because it is easier for nodes to cooperatively move toward the corresponding extreme. Instead, when entries are discordant, nodes have to mitigate between discordant opinions and thus are less free to vary in one of the two directions. This always occurs if the nodes are susceptible and non-stubborn, otherwise there would be disconnected communities and the cooperation would be impossible (this is what the conditions of the second part guarantee).

We can go one step further and provide (Theorems~\ref{theo:gFJ_p2p3_polarizingvectors}-\ref{theo:gFJ_p4_polarizingvectors} below) concrete cases of initial opinion vectors under which gFJ is polarizing with $P_2, P_3$, and  $P_4$. In the specific case of $P_2$ and $P_3$, we prove that finding the prejudice vector that yields maximum polarization, i.e. the maximum of function $\Delta_{P_{2,3}}$ defined in~\eqref{eq:delta}, is NP-hard, so we also discuss a possible approximation algorithm (Corollary~\ref{coro:gFJ_p2p3_v1_polarizingvectors}).


\begin{theorem}[gFJ: polarizing initial opinions for $P_2, P_3$]\label{theo:gFJ_p2p3_polarizingvectors}
Whenever the model is polarizing for $P_2, P_3$ (i.e. according to the conditions of Theorem~\ref{theo:gFJ_p2p3p4}), the polarizing prejudices $\bm{s}_{B_2(1)}, \bm{s}_{B_2(t)}, \bm{s}_{\max_{P_{2,3}}}$ can be derived as follows.
\renewcommand{\labelenumi}{(\roman{enumi})}
\begin{enumerate}
    \item Two polarizing prejudices $\pm\bm{s}_{B_2(1)}$ correspond to the unitary eigenvectors associated with the largest eigenvalue of matrix $H_g^TH_g$ and they correspond to the point of local maximum for the $P_2,P_3$-polarization on the $L_2$-ball of radius 1 $B_2(1) = \{ x \in [0,1]^n : \|x\|_2\leq 1\}$. In particular, it holds exactly $\Delta_{P_3}(\pm\bm{s}_{B_2(1)})=\sigma_1^2-1=\|H_g\|_2^2-1=|\mathcal{V}|\Delta_{P_2}(\pm\bm{s}_{B_2(1)})$, where $\sigma_1$ is the greatest singular value of the matrix $H_g$. Both these vectors have concordant entries.
    \item The opinion vectors $\pm\bm{s}_{B_2(t)}$ that yield the local maximum for $P_2,P_3$-polarization on the $L_2$-ball of radius $t$ $B_2(t)=\{ x \in [0,1]^n : \|x\|_2\leq t\}$ are given by $\pm\bm{s}_{B_2(t)}= \pm t \cdot \bm{s}_{B(1)}$, where $t= 1/s_{B_2(1)}^{(k)}$, with $s_{B(1)}^{(k)}$ denoting the largest entry of $\bm{s}_{B_2(1)}$. In particular, its polarization is exactly $t^2$ times the polarization of $\bm{s}_{B_2(1)}$. Both these vectors have concordant entries.
    %
    \item The global maximum for $P_2,P_3$-polarization is achieved for the initial opinion vectors $\pm \bm{s}_{\max_{P_{2,3}}} = \pm\sum_i \alpha_i \bm{v}_i $, whose components $\bm{\alpha}=(\alpha_i)$ can be obtained as the solution to the following optimization problem:
    \begin{align}
        \max & \qquad \sum_i \alpha_i^2 (\sigma_i^2-1)\nonumber\\
        \textrm{s.t.} & \qquad \bm{0} \leq B \bm{\alpha} \leq \bm{1} 
        \label{eq:optimization}
    \end{align}
    where $\sigma_1,\dots,\sigma_n$ are the singular values of $H_g$, $\bm{\alpha}=(\alpha_1,\dots,\alpha_n)^T$ is the vector of the coefficients that express $\textbf{s}_{max}$ with respect to the basis $\mathcal{B}=\{\textbf{v}_1\footnote{Please observe that $\bm{v}_1$ is the vector $\bm{s}_{B_2(1)}$ because it is the unitary eigenvector corresponding to the largest singular value},\dots,\textbf{v}_{n}\}$ composed of the unitary eigenvectors of $H_g^TH_g$, and $B$ is the matrix whose columns are the vectors of $\mathcal{B}$. The constraint guarantees that the solution $\bm{s}_{\max_{P_{2,3}}}$ has positive  (and $-\bm{s}_{\max_{P_{2,3}}}$ has respectively negative) is a proper opinion vector in $[-1,1]$ with concordant entries. This optimization problem, being quadratic non-convex programming, is NP-hard.
    \end{enumerate}
\end{theorem}

\noindent
Corollary~\ref{coro:gFJ_p2p3_v1_polarizingvectors} below tells us that, in case matrix $H_g$ has more than one singular value greater than one, it is possible to design sub-problems of the optimal problem described in~\eqref{eq:optimization} over spaces larger than $B_2(t)$ but smaller than the entire domain. These sub-problems are convex-quadratic programming and can be solved in polynomial time. Depending on the dimension of the network, numerical solutions may still not be found. Thus, we have designed a heuristic that always finds a solution $\pm\bm{s}_{V_{>1}}^{heu}$ whose polarization is greater than that of $\pm \bm{s}_{B_2(t)}$. The corresponding derivations can be found in the SI Appendix. 

\begin{corollary}[gFJ: polarizing initial opinions for $P_2, P_3$ on the subspaces $V_{>1}$ and $V_{\geq1}$]\label{coro:gFJ_p2p3_v1_polarizingvectors}
When matrix $H_g$ has more than one singular value greater than one, it is possible to design sub-problems of the optimal problem in~\eqref{eq:optimization} over $V_{>1}$ (vector space generated by the eigenvectors associated with the singular values strictly greater than $1$) and over $V_{\geq1}$ (vector space generated by the eigenvectors associated with the singular values greater or equal to $1$). These sub-problems yield polarizing vectors $\bm{s}_{V_{>1}}$, $\bm{s}_{V_{\geq1}}$, respectively, and they are convex-quadratic programming (with polynomial time complexity). A heuristic that always finds a solution $\pm\bm{s}_{V_{>1}}^{heu}$ is proposed. 
\end{corollary}

With Theorem~\ref{theo:gFJ_p2p3_polarizingvectors} and Corollary~\ref{coro:gFJ_p2p3_v1_polarizingvectors}, we are able to identify the initial opinion vectors $\pm\bm{s}_{\max_{P_{2,3}}}$, $\pm\bm{s}_{B_2(1)}$, $\pm\bm{s}_{B_2(t)}$, $\pm\bm{s}_{V>1}$, $\pm\bm{s}_{V\geq 1}$ leading to polarization maxima on the corresponding subspaces. While computing the opinion $\pm\bm{s}_{\max_{P_{2,3}}}$ yielding the global maximum is an NP-hard problem (Theorem~\ref{theo:gFJ_p2p3_polarizingvectors}.(iii)), an approximate solution could be obtained using standard numerical solvers (not in all cases, as we discuss in the Experimental Evaluation section). The \emph{local} polarization maxima are found reducing the problem on the subspaces corresponding to eigenvectors of $H_g^TH_g$ associated with singular values strictly greater or weakly greater than one. In particular, the vectors $\pm\bm{s}_{B_2(t)}$ of Theorem~\ref{theo:gFJ_p2p3_polarizingvectors}.(ii) (which are a scalar multiple of $\pm\bm{s}_{B_2(1)}$ in Theorem~\ref{theo:gFJ_p2p3_polarizingvectors}.(i)) are the vector that maximize the $P_2,P_3$-polarization on the space generated by  the eigenvectors $\pm\bm{s}_{B_2(1)}$ of $H_g^TH_g$ (also denoted with $\bm{v}_1$ in~Theorem~\ref{theo:gFJ_p2p3_polarizingvectors}.(iii)) that correspond to the singular value $\sigma_1>1$. The vectors $\pm\bm{s}_{V_{>1}}$ in Corollary~\ref{coro:gFJ_p2p3_v1_polarizingvectors} maximize the polarization on the larger subspace ${V_{>1}}$ generated by all the eigenvectors that correspond to the singular values strictly greater than one.  Finally, the vectors $\pm\bm{s}_{V_{\geq1}}$ in Corollary~\ref{coro:gFJ_p2p3_v1_polarizingvectors} maximize $P_2,P_3$-polarization on the even larger subspace ${V_{\geq1}}$ generated by all the eigenvectors that corresponds to the singular values weakly greater than one. Since these vectors correspond to the maximum of polarization over subspaces that are subset of each other, it is trivial to derive the following inequality:
\begin{align} \footnotesize
    \Delta_\Phi(\pm\bm{s}_{B_2(1)})&\leq\Delta_\Phi(\pm\bm{s}_{B_2(t)})\leq \Delta_\Phi(\pm\bm{s}_{V_{>1}}) \leq \nonumber\\
    &\leq 
    \Delta_\Phi(\pm\bm{s}_{V_{\geq 1}}) \leq \Delta_\Phi(\pm\bm{s}_{\max_{P_{2,3}}})
\label{eq:polarizing_inequality}
\end{align} 
for $\Phi=P_2,P_3$. 

While the results in Theorem~\ref{theo:gFJ_p2p3_polarizingvectors} and Corollary~\ref{coro:gFJ_p2p3_v1_polarizingvectors} do not have an immediate practical interpretation, we can get the gist of them with a simple numerical example. Consider a network composed of three nodes -- a naive node A, a node B with susceptibility value equal to 0.5, and a stubborn node C -- with  mutual weights equal to 0.5. Applying Theorem~\ref{theo:gFJ_p2p3_polarizingvectors}, we obtain that $\bm{s}_{B_2(1)}  = (0,0.30,0.95)$ and, dividing by $0.95$ as in Theorem~\ref{theo:gFJ_p2p3_polarizingvectors}.(ii), we obtain $\bm{s}_{B_2(t)}  = (0,0.31,1)$, which leads to a final opinion vector  $(0.8,0.61,1)$. The prejudice of the naive node A is opposite to that of stubborn node C, and A's opinion shifts significantly (from 0 to 0.8). The opinion of the intermediate node B is  approximately doubled.  The opinion vector achieving maximum polarization~$\bm{s}_{\textrm{max}_{P_{2,3}}}$ is instead $(0,0.75,1)$, whose corresponding final opinion is $(0.95,0.89,1)$. In this case, the combined effect of non-naive nodes' strong prejudices pushes A's final opinion to the opposite extreme. In some way, it is as if $\bm{s}_{B_2(t)}$ (which only takes into account one singular value of $H$) selected the prejudice that  maximizes the shift leveraging only to the most influential node (node~C). Instead, the $\bm{s}_{\textrm{max}_{P_{2,3}}}$ (which yields the global maximum) is able to enforce a synergy between non-naive nodes. In this simple case since $H$ has only one singular value greater than 1, we cannot obtain the vectors $\bm{s}_{V_{>1}}$ and $\bm{s}_{V_{\geq1}}$. 

\begin{theorem}[gFJ: polarizing vectors for $P_4$]\label{theo:gFJ_p4_polarizingvectors}
Whenever the model is polarizing for $P_4$ (i.e. according to the conditions of Theorem~\ref{theo:gFJ_p2p3p4}), the following hold true.
\renewcommand{\labelenumi}{(\roman{enumi})}
\begin{enumerate}
    \item Two prejudice vectors $\pm\bm{s}_{B_1(1)}$ that yields to $P_4$-polarization are the $j$-th vector of the standard basis in $\mathbb{R}^{n}$ (i.e., a vector whose components are all zero, except the $j$-th that equals 1) and its opposite, where $j = \underset{j}{\operatorname{argmax}} \sum_i h_{ij}$ (i.e., $j$ corresponds to the index of the column of $H_g = \{h_{ij}\}_{ij}$ with the greatest column-sum). This prejudice vector is also the point of maximum of $P_4$-polarization on the 1-norm ball $B_1(1)=\{x\in[0,1]^n: \|x\|_1\leq 1\}$ and its polarization is exactly given by $\Delta_{P_4}(\pm\bm{s}_{B_1(1)}))=\|H_g\|_1-1$.
    \item With the same notations of Theorem~\ref{theo:gFJ_p2p3_polarizingvectors}, the global maximum for $P_4$-polarization is achieved for the initial opinion vectors $\pm\bm{s}_{\max_{P_{4}}} = \pm\sum_i \alpha_i \bm{v}_i $ with concordant entries, whose components $\bm{\alpha}=(\alpha_1,\dots,\alpha_n)^T$ can be obtained as the solution to the following optimization problem:
    \begin{align}
        \max & \qquad \sum_i \alpha_i (\sigma_i^2-1)\langle \bm{v}_i,\bm{1}\rangle\nonumber\\
        \textrm{s.t.} & \qquad \bm{0} \leq B \bm{\alpha} \leq \bm{1} 
        \label{eq:optimization_p4}
    \end{align}
    This optimization problem is a linear programming problem that can be numerically solved.
\end{enumerate}
\end{theorem}

\noindent
As observed for $P_2,P_3$-polarization, it is holds that:
\begin{equation}
    P_4(\pm\bm{s}_{B_1(1)})\leq P_4(\pm\bm{s}_{\max_{P_4})}).
    \label{eq:polarizing_inequality_p4}
\end{equation}



\subsection{Polarization under $P_1$ and $GDI$}
We conclude the analysis of gFJ by studying the polarization under $P_1$ and $GDI$. For this case, Theorem~\ref{theo:gFJ_p1gdi} asserts that whenever gFJ does not polarize in $P_2,P_3,P_4$, it does not polarize in $P_1,GDI$ either. Instead, when gFJ is polarizing in $P_2,P_3,P_4$, we can guarantee that it also polarizes in $P_1,GDI$ only if the sufficient condition in Theorem~\ref{theo:gFJ_p1gdi} is satisfied. Again, the proof of the theorem below can be found in the SI Appendix.

\begin{theorem}[gFJ: global polarization with $P_1, GDI$]\label{theo:gFJ_p1gdi}
For polarization indices $P_1$ and $GDI$, the following results hold:
\renewcommand{\labelenumi}{(\roman{enumi})}
\begin{enumerate}
\item if gFJ is depolarizing for $P_2,P_3,P_4$, then it is also depolarising for $P_1,GDI$;
\item gFJ is polarizing in $\bm{s}$ if the following condition holds true: 
\begin{align}
\sum_i\alpha_i^2(\sigma_i^2-1) \geq \frac{1}{n}&\left[\sum_i|\alpha_i|(\sigma_i^2-1)\langle|\bm{v}_i|,\bm{1}\rangle\right]  \cdot   \nonumber\\
&\cdot\left[\sum_i|\alpha_i|(\sigma_i^2+1)\langle|\bm{v}_i|,\bm{1}\rangle\right],
   \label{eq:p1gdi_cond_general}
\end{align}
where $\bm{\alpha}=(\alpha_1,\dots,\alpha_n)^T$ is the expression of $\bm{s}$ in terms of the basis $\mathcal{B}$ of the unitary eigenvectors of $H_g^TH_g$;
%
\end{enumerate}
\end{theorem}

\subsection{The role of stubborn and naive nodes}

We now show (Corollary~\ref{coro:gFJ_stubborn} below, proof in SI Appendix) a general result regarding stubborn nodes (i.e., nodes whose opinion is not at all swayed by that of their peers, which translates into $\lambda_i = 0$), whose role has not a direct impact on polarization. In fact, we will see that even if their strong \emph{anchoring} attitude would intuitively suggest that they always have an effect on the final opinion, the network structure can instead invalidate it.

\begin{corollary}\label{coro:gFJ_stubborn}
While naive nodes tend to make polarization easier, stubborn nodes do not have a clear directional effect on the polarization with $P_2, P_3, P_4$. 
\end{corollary}

Leveraging Theorems~\ref{theo:gFJ_p2p3p4} and~\ref{theo:gFJ_p1gdi}, we can also study a special case involving naive nodes. This result, whose proof can be found in the SI Appendix of this paper, emphasizes the role of naive nodes (the ones with $\lambda_i$=1), which essentially forget their prejudice and move their opinion towards the opinion of the other nodes.

\begin{corollary}\label{coro:gFJ_naive}
Let us assume that the set of nodes $\mathcal{V}$ is composed of two disjoint groups, $\mathcal{I}$ and $\mathcal{J}$, such that all non-naive nodes have the same opinion $\tau$, while the naive nodes' opinions are free to vary in $[-1,1]$, or equivalently:
\begin{equation}
    \forall i \in \mathcal{I} \quad \lambda_i = 1, s_i \in [-1,1] \qquad \forall j \in \mathcal{J} \quad \lambda_j <1, s_j = \tau.
\end{equation}
Then, the final opinion $z$ is exactly the vector $\bm{z}=\tau\mathbf{1}$. 
In addition, this configuration is never polarizing for $P_1$ and $GDI$, while, as long as $|s_i| < 1$ for at least one node $i$, it always exists a $\tau$ value such that $P_2$, $P_3$, and $P_4$ are polarizing. 
\end{corollary}

\vspace{-15pt}
\section{vFJ polarizes when gFJ does}
\label{sec:vFJ_all}

As already observed, in vFJ the opinion of a node $i$ at step $k$ does not take into account its own opinion at step $k-1$ (as it happens, instead, with gFJ, which weighs it with $w_{ii}$). Thus, from a mathematical standpoint, the two models are different. However, apart from this specific contribution (i.e. in the case $w_{ii}$ is null), vFJ can be manipulated to exactly yield the same polarization as gFJ, if in an indirect and less intuitive way. In fact, in gFJ the susceptibility parameter directly captures the innate tendency of a node to be influenced (and to which degree) by others. In vFJ, instead, the rate at which a node is influenced by its peers is captured by: (i) the social strength of the node with all its neighbours $\hat{w}_{ij}$, (ii) the anchoring-degree of the node itself $\hat{w}_{ii}$, i.e. the importance it assigns to its initial prejudice. 

Theorem~\ref{theo:vFJ} below establishes a complete equivalence, in terms of polarization properties, between gFJ and vFJ. 


\begin{theorem}[vFJ: local and global polarization]\label{theo:vFJ}
For all polarization metrics, the vFJ model yields polarization under exactly the same conditions as gFJ. Specifically, if we replace matrix $H_g$ with the vFJ matrix $H_v$ and we set $\hat{w}_{ii}=0$ for naive nodes (if present), the results of Theorems~\ref{theo:gFJ_p2p3p4}-\ref{theo:gFJ_p1gdi} and Corollaries~\ref{theo:gFJ_p2p3_polarizingvectors}-\ref{theo:gFJ_p4_polarizingvectors} hold true. In particular, the condition for $H_v$ not being  doubly stochastic reduces from~\eqref{eq:stochasticity_cond}  to the following one:
\begin{equation}
    \frac{\sum_{j\neq i}\hat{w}_{ij}}{\hat{w}_{ii}}-\sum_{j\neq i}\frac{\hat{w}_{ji}}{\hat{w}_{jj}} \neq 0.
    \label{eq:condition_vFJ_p2p3}
\end{equation}
\end{theorem}
\begin{proof} The proof consists in the derivation of vFJ from gFJ. This can be done using the following mapping:
\begin{align}
    &\hat{w}_{ii}\neq\infty\rightarrow     \left\{
    \begin{array}{l}
    \lambda_i=\frac{\sum_{k\in\mathcal{N}(i)}\hat{w}_{ik}}{\hat{w}_{ii}+\sum_{k\in\mathcal{N}(i)}\hat{w}_{ik}}\\
    w_{ij}=\frac{\hat{w}_{ij}}{\sum_{k\in\mathcal{N}(i)}\hat{w}_{ik}}
    \end{array} \right.\label{eq:mapp1}\\
    &\hat{w}_{ii}=\infty\rightarrow     \left\{
    \begin{array}{l}
    \lambda_i=0\\
    w_{ij}=0 \label{eq:mapp2}
    \end{array}\right.
\end{align}
Thus, the thesis follows from the results obtained for gFJ.
\end{proof}
\eqref{eq:condition_vFJ_p2p3} simplifies  when the social graph is undirected (which corresponds to the matrix $\hat{W}$ being symmetric). As stated in Corollary~\ref{coro:vFJ_p2p3p4_undirected} below, in that case, when the self-weights are identical for all nodes (i.e., $\hat{w}_{ii} = \hat{w}, \forall i$) vFJ is \emph{never} polarising in any metric and the average opinion is invariant to the opinion formation process.


\begin{corollary}[vFJ on undirected social graphs]\label{coro:vFJ_p2p3p4_undirected}
When the social graph is undirected (i.e. matrix $\hat{W}$ is symmetric), vFJ is polarizing with $P_2$, $P_3$ and $P_4$ if and only if $\hat{w}_{ii}$ are not identical for all $i$. When $\hat{w}_{ii} = \hat{w}, \forall i$, vFJ is \emph{never} polarising in any metric and it holds that $\sum_i z_i=\sum_i s_i$, i.e., there is never a choice shift in the network and the average final opinion is the same as the average initial opinion.
%
\end{corollary}
\vspace{-10pt}
\section{The rFJ model is never polarizing in undirect networks}

In this section, we derive the results of polarization on the rFJ model. We already know from Bindel \etaltwo~\cite{Bindel2015}  that rFJ does not polarize according to the local definition $NDI$, and, from Gionis \etaltwo~\cite{Gionis2013}, that in the specific case of undirected social graph it does not polarize according to the global definition $P_4$. Here, we generalize these findings. To this aim, note that rFJ is equivalent to vFJ after setting $w_{ii}=1$. Thus, Theorem~\ref{theo:vFJ} also applies in this case. The condition for $H_r$ (the equivalent of $H_g$ but for rFJ) not being doubly stochastic, simply reduces from~\eqref{eq:condition_vFJ_p2p3} to the following one:
\begin{equation}
    \sum_{j\neq i}\hat{w}_{ij}-\sum_{j\neq i}\hat{w}_{ji} \neq 0.
    \label{eq:condition_rFJ_p2p3}
\end{equation}
And when the social graph is undirected, we obtain an even stronger result, summarized in Corollary~\ref{coro:rFJ} below.

\begin{corollary}[rFJ on undirected social graphs]\label{coro:rFJ}
The rFJ model is never polarizing, in any polarization metrics, for any initial opinion vector. In addition, it holds that $\sum_i z_i=\sum_i s_i$, i.e., there is never a choice shift in the network and the average final opinion is the same as the average initial opinion.
\end{corollary}

\emph{Remark:} while polarization was still possible under vFJ on undirected social graph, with rFJ polarization never happens. 
The practical implication of this result for undirected social graphs is that polarization (in all its variations) can never be induced ``naturally'' by an opinion formation process following rFJ. 
Even more interestingly, polarization cannot be induced by altering the social graph (as long as it stays symmetric). Thus, when an initial state $s$ is given, the final state $z$ with rFJ can only naturally evolve towards non-polarization. 
Vice versa, when the social graph is directed, the above result does not hold since, in a directed graph, the opinion of nodes with stronger social power tends to steer the opinion of the others. While relationship-oriented online social networks, like Facebook, tend to feature undirected graphs,
directed social graphs are common in information-driven online social networks like Twitter.  

\vspace{-10pt}
\section{Experimental evaluation}
\label{sec:evaluation}

In this section we analyze the theoretical results on two real social network graphs: the Karate Club graph~\cite{zachary1977information} and a Facebook graph~\cite{leskovec2012learning}. The Karate Club dataset corresponds to an unweighted graph composed of 34 members. The Facebook dataset is a Facebook snapshot comprising 4039 users. Also in this case the graph is unweighted. After discarding isolated nodes (since they do not contribute at all to the opinion formation process), we end up with a network of 1519 nodes. With these datasets we obtain the values $\hat{w}_{ij}$ that describe the social links between different users. 

For both graphs, we obtain the influence matrix from the social matrix $\hat{W}$ normalizing by rows, i.e. $w_{ij}=\frac{\hat{w}_{ij}}{\sum_k \hat{w}_{ik}}$. To proceed with the analysis we should set the  susceptibility values of nodes, which are not fixed by the social network. To this aim, since both networks have a few very central nodes, as displayed in the SI Appendix, we decided to use a centrality measure to set them. In the following, we will show the results obtained considering the PageRank centrality, which is the centrality measure that better captures the influence among nodes~\cite{proskurnikov2017tutorial}, but similar results hold with other centrality measures (betweenness, degree, eigenvector and k-shell centrality). In our experiments, if the Pagerank centrality of a node $i$ is $\mathrm{C}_i$,  we assign $\lambda_i$ the value of $\mathrm{C}_i$  (and $\mathrm{C}_i^{-1}$) rescaled to $(0,1)$, so that the more central the nodes (and, respectively, the less central), the higher their susceptibility values. Furthermore, we will also show the case in which all nodes have the same susceptibility, set to 0.8. In the SI Appendix, we provide a visualization of the social networks we consider and of the susceptibility values obtained in this way for both datasets.

\newcommand{\gc}{\cellcolor[HTML]{EFEFEF}}

\begin{table}
\centering
\caption{gFJ in the Karate network: values of $\Delta_\Phi$ for all polarization metrics, for the three $\lambda_i$ configurations. The shadowed area highlights were the corresponding opinion vectors $\bm{s}$ are expected to yield polarization.}\label{tab:karate_s}
\setlength{\tabcolsep}{2.5pt}
\begin{tabular}{@{}c cccccc@{}}
\multicolumn{7}{c}{$\lambda_i \propto \mathrm{P}_i$ } \\ 
\toprule
  & $\Delta_{P_1}$ & $\Delta_{P_2}$ & $\Delta_{P_3}$ & $\Delta_{P_4}$ & $\Delta_{NDI}$ & $\Delta_{GDI}$ \\
  \midrule
 $\bm{s}_{unif}$ &  $-8.8\ex{-1}$ & $-1.7\ex{-2}$ & $-5.7\mathrm{e}{-1}$ & $3.1\ex{-1}$ & $-1.03$ & $-29.93$ \\
 $\bm{s}_{B_2(1)}$ & $-6.4\ex{-2}$ & \gc $2.2\ex{-3}$ & \gc $7.6\ex{-2}$ & $4.2\ex{-1}$ & $-3.1\ex{-1}$ & $-2.18$ \\
 $\bm{s}_{B_2(t)}$ & $-1.15$ & \gc $4.0\ex{-2}$ & \gc $1.36$ & $1.77$ & $-5.61$ & $-38.94$ \\
 $\bm{s}_{V_{>1}}$ & $-1.18$ & \gc $4.1\ex{-2}$ & \gc $1.40$ & $1.80$ & $-5.67$ & $-40.28$ \\
 $\bm{s}_{V_{>1}}^{heu}$ & $-1.18$ & \gc $4.1\ex{-2}$ & \gc $1.40$ & $1.80$ & $-5.67$ & $-40.28$ \\
 $\bm{s}_{\max_{P_{2,3}}}$  & $-1.81$ & \gc $6.0\ex{-2}$ & \gc $2.05$ & $2.05$ & $-6.94$ & $-61.48$ \\
 $\bm{s}_{B_1(1)}$ & $-2.0\ex{-1}$ & $-5.4\ex{-3}$ & $-1.9\ex{-1}$ & \gc $1.6\ex{-1}$ & $-1.3\ex{-1}$ & $-6.63$ \\
 $\bm{s}_{\max_{P_4}}$  & $-3.57$ & $3.4\ex{-2}$ & $1.16$ & \gc $2.65$ & $-10.65$ & $-121.421$ \\
\bottomrule
\multicolumn{7}{c}{$\lambda_i \propto \mathrm{P}_i^{-1}$ } \\
\hline
 $\bm{s}_{unif}$ & $-1.62$ & $-4.8\ex{-2}$ & $-1.63$ & $-1.1\ex{-2}$ & $-1.41$ & $-55.22$ \\
 $\bm{s}_{B_2(1)}$ & $6.8\ex{-1}$ & \gc $4.3\ex{-2}$ & \gc $1.45$ & $3.32$ & $-3.1\ex{-1}$ & $23.19$ \\
 $\bm{s}_{B_2(t)}$ & $7.6\ex{-1}$ & \gc $4.3\ex{-2}$ & \gc $1.63$ & $3.52$ & $-3.5\ex{-1}$ & $25.97$ \\
 $\bm{s}_{V_{>1}}$ & $3.0\ex{-1}$ & \gc $1.3\ex{-1}$ & \gc $4.48$ & $7.50$ & $-5.08$ & $10.07$ \\
 $\bm{s}_{V_{>1}}^{heu}$ & $3.1\ex{-1}$ & \gc $1.1\ex{-1}$ & \gc $3.61$ & $6.74$ & $-3.72$ & $10.45$ \\
 $\bm{s}_{\max_{P_{2,3}}}$  & $-2.30$ & \gc $2.01\ex{-1}$ & \gc $6.86$ & $6.86$ & $-3.98$ & $-78.17$ \\
 $\bm{s}_{B_1(1)}$ & $3.1\ex{-1}$  & $2.2\ex{-2}$ & $7.5\ex{-1}$ & \gc $2.98$ & $-1.95$ & $10.67$ \\
 $\bm{s}_{\max_{P_4}}$  & $-4.27$ & $1.3\ex{-1}$ & $4.59$ & \gc $10.04$ & $-9.36$ & $-145.22$ \\
 \arrayrulecolor{black}\bottomrule
\multicolumn{7}{c}{$\lambda_i = 0.8$ } \\
\hline
 $\bm{s}_{unif}$ & $-2.78$ & $-7.7\ex{-2}$ & $-2.62$ & $1.7\ex{-1}$ & $-2.45$ & $-94.67$ \\
 $\bm{s}_{B_2(1)}$ & $-3.3\ex{-1}$ & \gc $1.5\ex{-1}$ & \gc $5.2\ex{-1}$ & $2.44$ & $-1.22$ & $-11.30$ \\
 $\bm{s}_{B_2(t)}$ & $-1.43$ & \gc $6.6\ex{-2}$ & \gc $2.23$ & $5.06$ & $-5.22$ & $-48.49$ \\
 $\bm{s}_{\max_{P_{2,3}}}$  & $-1.73$ & \gc $1.3\ex{-1}$ & \gc $4.57$ & $4.57$ & $-3.33$ & $-58.73$ \\
 $\bm{s}_{B_1(1)}$ & $-8.2\ex{-1}$ & $-1.5\ex{-2}$ & $-5.2\ex{-1}$ & \gc $2.34$ &  $-5.60$ & $-27.90$  \\
 $\bm{s}_{\max_{P_4}}$  & $-6.86$ & $-1.7\ex{-1}$ & $-1.7\ex{-1}$ & \gc $8.10$ & $-13.32$ & $-233.38$ \\
 \hline
\end{tabular}
\vspace{-10pt}
\end{table}

We can now search for the initial opinion vectors that yield polarization in the social network, by applying Theorems~\ref{theo:gFJ_p2p3_polarizingvectors}-\ref{theo:gFJ_p4_polarizingvectors} and Corollary~\ref{coro:gFJ_p2p3_v1_polarizingvectors}. For the sake of brevity, in the following we will consider only the positive polarizing vectors but analogous results can be obtained for negative ones, as stated in the corollaries.
We compute the polarizing vectors $\bm{s}_{B_2(1)}$, $\bm{s}_{B_2(t)}$, $\bm{s}_{\max_{P_{2,3}}}$, $\bm{s}_{V>1}$, $\bm{s}_{V>1}^{heu}$, $\bm{s}_{B_1(1)}$, $\bm{s}_{\max_{P_4}}$ as described in Theorems~\ref{theo:gFJ_p2p3_polarizingvectors}-\ref{theo:gFJ_p4_polarizingvectors} and Corollary~\ref{coro:gFJ_p2p3_v1_polarizingvectors}, and we  compare their polarization with the one of an opinion vector $\bm{s}_{unif}$ with entries randomly drawn from a uniform distribution in $[0,1]$. 
Table~\ref{tab:karate_s} shows the polarization induced by the above vectors on the Karate social network, for the three susceptibility configurations we are considering.  
Recall that the polarization shift $\Delta_\Phi(s)$ for a given polarization metric $\Phi$ (with $\Phi=P_1,\dots,P_4,NDI,GDI$) and initial opinion $\bm{s}$ is derived as $\Phi(Hs)-\Phi(s)$. When $\Delta_\Phi(s)$ is positive, then gFJ polarizes in $s$.
We can see in Table~\ref{tab:karate_s} that the theoretical results are confirmed (this is not surprising, since our theorems are obtained without any approximation). A random prejudice vector $\bm{s}_{unif}$ leads to depolarization for all the polarization metrics. Instead, the prejudices from Theorem~\ref{theo:gFJ_p2p3_polarizingvectors} and Corollary~\ref{coro:gFJ_p2p3_v1_polarizingvectors} yield to $P_2,P_3$-polarization. As expected, according to \eqref{eq:polarizing_inequality}, their corresponding $P_2,P_3$-polarization shifts are progressively increasing moving from $\bm{s}_{B_2(1)}$ to $\bm{s}_{\max_{P_{2,3}}}$ (because the solution is searched for into a larger domain). 
Note that, since the network is small, the numerical solver is able to find the solutions $\bm{s}_{\max_{P_{2,3}}}$ and $\bm{s}_{V_{>1}}$ (the latter is not applicable to the case $\lambda_i =0.8$, because its $H_g$ has only one singular value greater than 1). 
It is interesting to observe that  the solution $\bm{s}_{V_{>1}}^{heu}$ found with the heuristic is, in one case, exactly equal to  the one obtained numerically ($\bm{s}_{V_{>1}}$) and, in the other case, extremely close to it, which confirms the heuristic validity.
The prejudice vectors found according to  Theorem~\ref{theo:gFJ_p4_polarizingvectors}, instead,  yield to $P_4$-polarization, and satisfy the inequality in \eqref{eq:polarizing_inequality_p4}.
With respect to $P_1,GDI$-polarization, while Theorem~\ref{theo:gFJ_p1gdi} cannot tell us whether polarization is achieved in general, we can use it to predict whether $P_1,GDI$-polarization is achieved with the same prejudices that yield $P_2,P_3$ or $P_4$ polarization. We find that the condition (sufficient for polarization) of Theorem~\ref{theo:gFJ_p1gdi} is verified only for the $P_2,P_3$-polarizing prejudices and $\lambda_i \propto \mathrm{C}_i^{-1}$. The columns $\Delta_{P_1}$ and $\Delta_{GDI}$ of Table~\ref{tab:karate_s} confirm polarization in these cases. Finally, as expected from Theorem~\ref{theo:gFJ_ndi}, gFJ is always depolarizing in $NDI$.

Similar results are obtained with the Facebook network (Table ~\ref{tab:facebook_s}). Two points are worth emphasizing. First, the centrality of nodes in the Facebook graph is extremely skewed, with one very central node dominating the graph. Thus, when $\lambda_i \propto \mathrm{C}_i$, there are very few susceptible nodes and polarization is harder to achieve. The opposite effect is observed when $\lambda_i \propto \mathrm{C}_i^{-1}$, and the polarization shifts are higher. Second, note that 
since the Facebook network size is large, the global solutions ($s_{max_{P_{2,3}}}$ and $s_{max_{P_4}}$) could not be found numerically and $\bm{s}_{V_{>1}}$ could only be obtained for $\lambda_i\propto \mathrm{C}_i$. This example showcases the importance of the heuristics derived in the previous section, which can always return a polarizing vector.

\begin{table}
\centering
\caption{gFJ in the Facebook network: values of $\Delta_\Phi$ for all polarization metrics, for the three $\lambda_i$ configurations. The shadowed area highlights were the corresponding opinion vectors are expected to yield polarization.}\label{tab:facebook_s}
\setlength{\tabcolsep}{2.5pt}
\begin{tabular}{@{}c cccccc@{}}
\multicolumn{7}{c}{$\lambda_i \propto \mathrm{C}_i$ } \\
\toprule
  & $\Delta_{P_1}$ & $\Delta_{P_2}$ & $\Delta_{P_3}$ & $\Delta_{P_4}$ & $\Delta_{NDI}$ & $\Delta_{GDI}$ \\
  \midrule
 $\bm{s}_{unif}$ &  $-6.2\ex{-1}$ & $-1.7\ex{-4}$ & $-2.6\ex{-1}$ & $3.6\ex{-1}$ & $-0.31$ & $-940$ \\
 $\bm{s}_{B_2(1)}$ & $-3.6\ex{-4}$ & \gc $5.5\ex{-7}$ & \gc $8.4\ex{-4}$ & $2.8\ex{-2}$ & $-8.8\ex{-4}$ & $-5.5\ex{-1}$ \\
 $\bm{s}_{B_2(t)}$ & $-3.8\ex{-1}$ & \gc $5.8\ex{-4}$ & \gc $8.7\ex{-1}$ & $8.9\ex{-1}$ & $-9.2\ex{-1}$ & $-557$ \\
 $\bm{s}_{V_{>1}}$ & $-4.1\ex{-1}$ & \gc $5.8\ex{-4}$ & \gc $8.8\ex{-1}$ & $9.0\ex{-1}$ & $-9.3\ex{-1}$ & $-630$ \\
 $\bm{s}_{V_{>1}}^{heu}$ & $-4.1\ex{-1}$ & \gc $5.8\ex{-4}$ & \gc $8.8\ex{-1}$ & $9.0\ex{-1}$ & $-9.3\ex{-1}$ & $-630$ \\
 $\bm{s}_{B_1(1)}$ & $-1.6\ex{-1}$ & $-1.0\ex{-4}$ & $-1.6\ex{-1}$ & \gc $6.1\ex{-2}$ & $-2.7\ex{-1}$ & $-240$ \\
\bottomrule
\multicolumn{7}{c}{$\lambda_i \propto \mathrm{C}_i^{-1}$ } \\
\hline
 $\bm{s}_{unif}$ & $-64.51$ & $-1.2\ex{-1}$ & $-180.64$ & $-129.81$ & $-9.45$ & $-97985$ \\
 $\bm{s}_{B_2(1)}$ & $53.83$ & \gc $1.1\ex{-1}$ & \gc $171.10$ & $419.61$ & $-6.3\ex{-2}$ & $81774$ \\
 $\bm{s}_{B_2(t)}$ & $53.94$ & \gc $1.1\ex{-1}$ & \gc $171.45$ & $420.04$ & $-6.3\ex{-1}$ & $81941.1$ \\
 $\bm{s}_{V_{>1}}^{heu}$ & $49.80$ & \gc $1.34\ex{-1}$ & \gc $203.59$ & $497.04$ & $-1.35$ & $75646$ \\
  $\bm{s}_{B_1(1)}$ & $53.73$  & $1.1\ex{-1}$ & $170.75$ & \gc $420.62$ & $-6.3\ex{-1}$ & $81615$ \\
 \bottomrule
\multicolumn{7}{c}{$\lambda_i = 0.8$ } \\
\hline
 $\bm{s}_{unif}$ & $-109.77$ & $-1.6\ex{-1}$ & $-243.81$ & $152.36$ & $-22.71$ & $-166737$ \\
 $\bm{s}_{B_2(1)}$ & $60.34$ & \gc $1.3\ex{-1}$ & \gc $198.13$ & $455.2$ & $-1.02$ & $91656$ \\
 $\bm{s}_{B_2(t)}$ & $60.44$ & \gc $1.3\ex{-1}$ & \gc $198.44$ & $455.56$ & $-1.02$ & $91693$ \\
 $\bm{s}_{V_{>1}}^{heu}$ & $25.47$ & \gc $1.4\ex{-1}$ & \gc $217.38$ & $523.82$ & $-4.03$ & $38687$ \\
 $\bm{s}_{B_1(1)}$ & $60.36$ & $1.3\ex{-1}$ & $197.81$ & \gc $455.93$ &  $-1.03$ & $91693$  \\
 \hline
\end{tabular}
\vspace{-15pt}
\end{table}

We conclude this section by having a closer look at how polarizing prejudices are structured. 
In Figure~\ref{fig:results_Karate}, each arrow corresponds to one node in the Karate graph, and it starts at its prejudice and ends at its final opinion. For $P_2,P_3,P_4$, an increase in polarization is linked, intuitively, to some opinions moving from  more neutral states (close to 0) to more extreme states (close to 1). Indeed, this is what happens in all the cases presented in the figure. In particular, the vectors $\bm{s}_{\textrm{max}_{P_{2,3}}}$ and $\bm{s}_{\textrm{max}_{P_4}}$ that maximize the polarization feature the maximum number of components with initial opinion equal to 1 (with respect to the other opinion vectors): in this way, the nodes with more extreme opinions work synergistically to push the others' opinions closer to theirs. For selecting such an optimal ``cooperative'' group of extreme nodes, one should be able to search for a solution to the optimization problem within the entire domain of opinions. When this is not the case, only suboptimal polarization is achieved.
For example, the vectors $\bm{s}_{B_2(1)},\bm{s}_{B_2(t)}, \bm{s}_{B_1(1)}$, only manage to select one single extreme node responsible for pushing the more neutral opinions of others, while $\bm{s}_{V>1}$ is in an intermediate position, being able to select more extreme nodes than  $\bm{s}_{B_2(t)}$ and fewer than $\bm{s}_{\textrm{max}_{P_{2,3}}}$.
We can also observe that in panels A and B of Figure~\ref{fig:results_Karate}, where the susceptibility varies across nodes, the nodes with initial opinion 1 are always the most stubborn, so that they create a field of attraction for more susceptible nodes. Effectively, the susceptibility assigned to nodes overrides their centrality in the network, hence very central nodes can become attractors or attractees depending on how stubborn they are. Vice versa, when the susceptibility of all nodes is the same (panel C of Figure~\ref{fig:results_Karate}), we observe the unfiltered effect of centrality: the most polarizing prejudices are those in which the most central nodes have initial opinions close to 1, and their final opinion changes much less than the others' opinions. This also confirms that the PageRank centrality is able to capture the ability of nodes to convince the others, and thus it identifies the most influential nodes. 

In Figure~\ref{fig:results_Facebook} we can see the results obtained with the Facebook network. In this case, since the network is large, it is not possible to find the global solutions $\bm{s}_{\max_{P_{2,3}}}$ and $\bm{s}_{\max_{P_{4}}}$. However, the considerations we made for the Karate graph hold also in this case. In particular, the polarizing vectors assign to more stubborn nodes initial opinions closer to 1, so that they can influence susceptible nodes to which they are connected. In the Facebook network, though, due to the scale-free topology with just a few  hubs and many poorly connected nodes, we also observe very susceptible nodes that do not change their opinions (Figure~\ref{fig:results_Facebook}, panel B). These nodes have typically a single edge towards a stubborn node sharing its opinion.



\begin{figure*}[p]
\centering
    \begin{subfigure}[b]{\textwidth}
        \centering
        \includegraphics[width=0.9\textwidth]{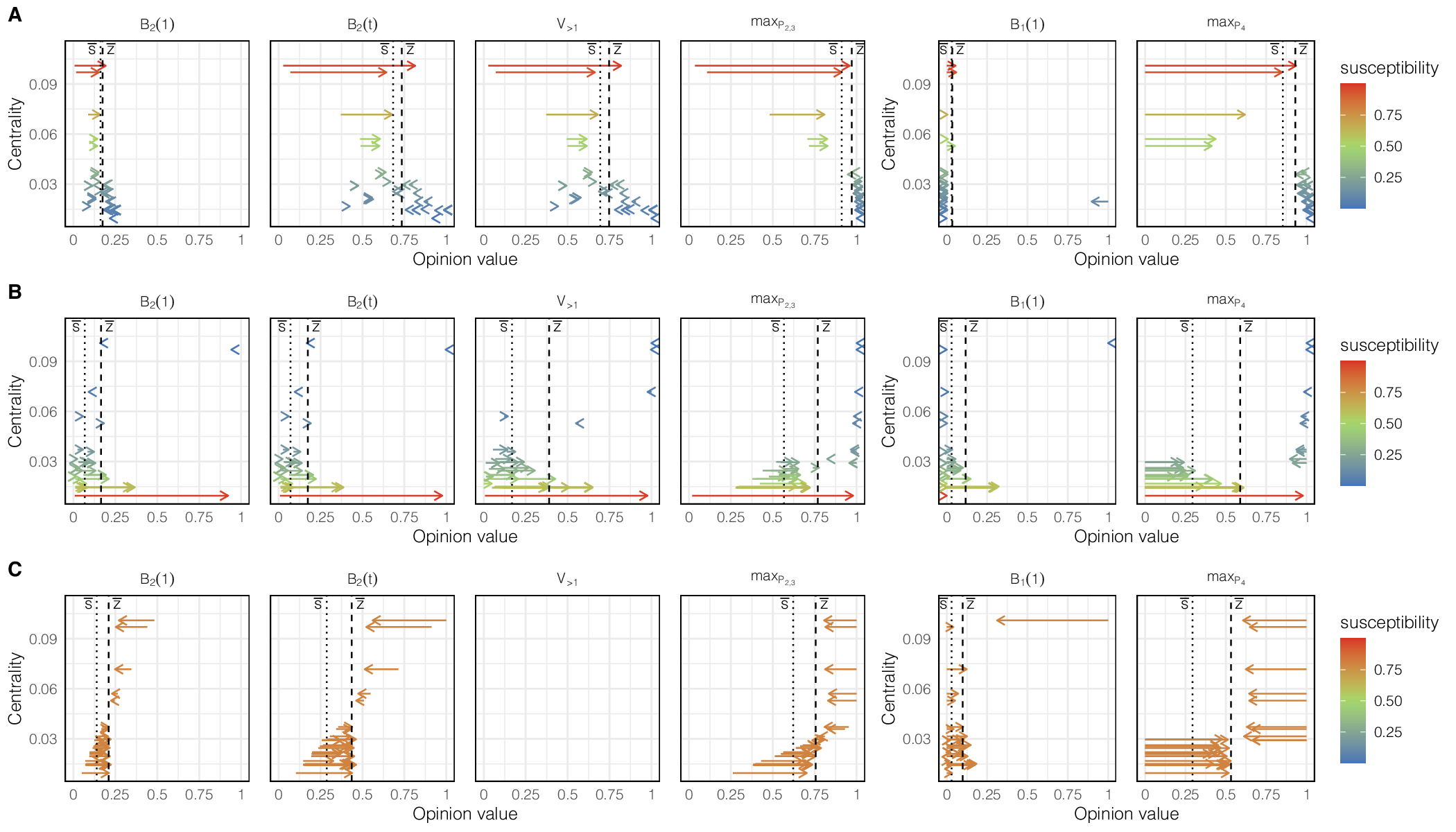}
        \caption{gFJ on the Karate network.}
        \label{fig:results_Karate}
    \end{subfigure}
    \begin{subfigure}[b]{\textwidth}
        \centering
        \includegraphics[width=0.9\textwidth]{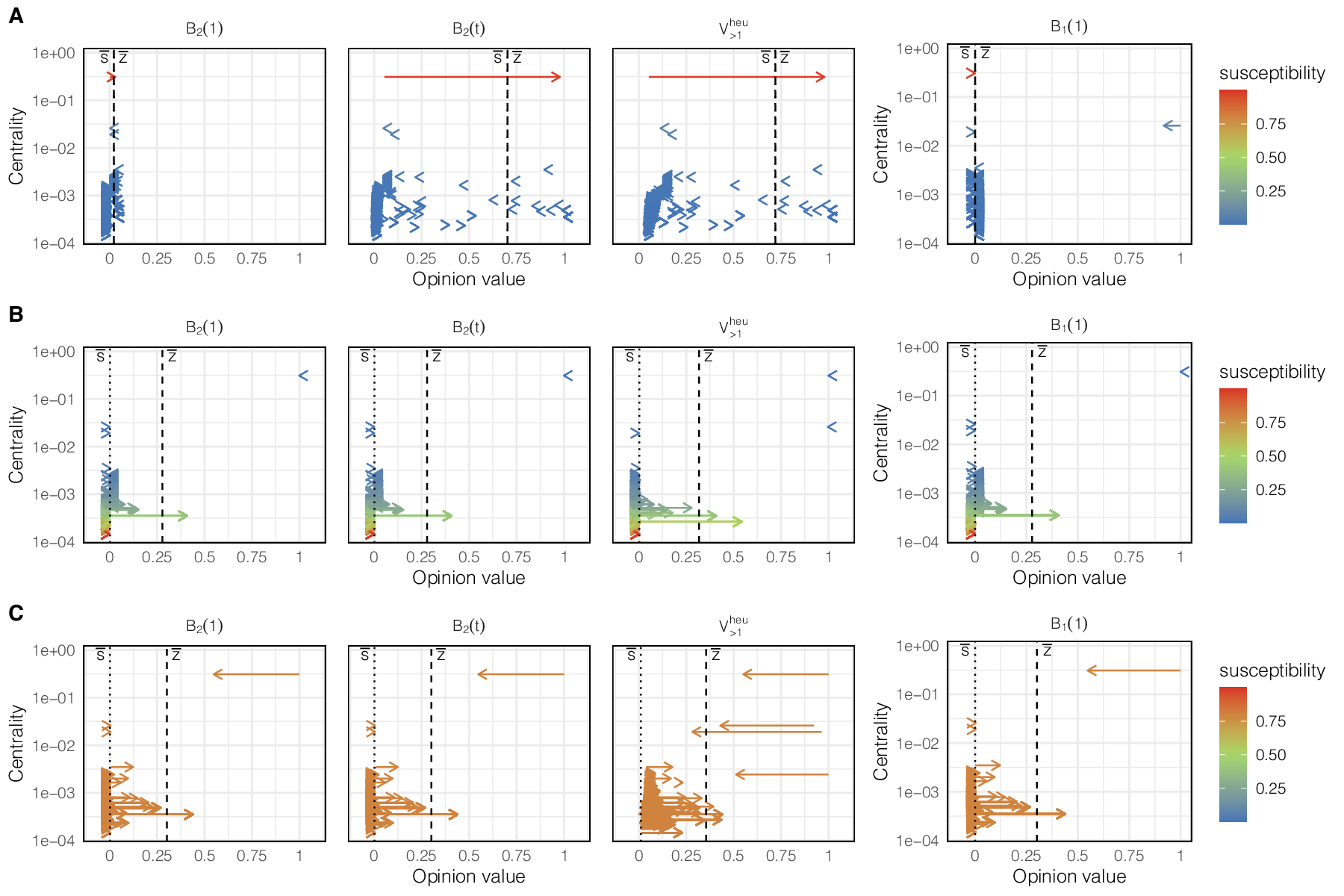}
        \caption{gFJ on the Facebook network.}
        \label{fig:results_Facebook}
    \end{subfigure}
\vspace{-10pt}
 \caption{Polarizing opinions vs nodes centrality. Each arrow corresponds to one node in the graph. An arrow starts at the initial prejudice and ends at the final opinion. The color of the arrow corresponds to the susceptibility assigned to the node. (A) $\lambda_i \propto \mathrm{C}_i$ (B) $\lambda_i \propto \mathrm{C}_i^{-1}$ (C) $\lambda_i = 0.8$. In each panel, on the left the opinions yielding $P_2,P_3$-polarization, on the right the opinions yielding $P_4$-polarization. With the dotted and dashed line we denote the average initial and final opinion, respectively.}
 \label{fig:results}
 \vspace{-10pt}
\end{figure*}

\vspace{-20pt}
\section{Conclusions}
In this work we have investigated under which conditions the popular Friedkin-Johnsen model yields polarised opinions. The first contribution of the work has been to systematize the variety of FJ models used in the literature, and the many definitions of polarization. 
Then, as the main contribution of the work, we have derived the conditions under which the FJ models yield to polarization, for each of the polarization classes identified from the related literature. Moreover, we have identified a methodology for obtaining polarizing prejudices in most cases. When exact solutions could not be found (because the corresponding problem was NP-hard), we have defined heuristics to find a sub-optimal solution. Our theoretical results have then been tested on two real-life social networks. We have seen that both the centrality of nodes in the social network as well as their individual susceptibility to the opinions of other nodes play a key role in defining their influence power, hence their ability to polarize. 

The results presented in this work can be used to understand under which conditions polarization of opinions will emerge for a given social network. While the application to online social networks immediately comes to mind (as showcased in Section~\ref{sec:evaluation}, the social graph can be collected from online social network platforms such as Twitter, Facebook, Reddit, etc), other applications can be foreseen, such as failure mode and effect analysis in reliability engineering~\cite{ZHANG2021}. In addition, the results presented in this paper can be exploited to design interventions to bring polarization under control. More in general, since opinions in the FJ model are actually abstracted as values in the $[0,1]$ or $[-1,1]$ domain, the FJ model could be used to study information propagation, the evolution of decision processes, and consensus/polarization on networks, as long as the mapping in the same unidimensional domain remains appropriate.

\vspace{-15pt}

\bibliography{references.bib}{}
\bibliographystyle{IEEEtran}

\vskip -2\baselineskip plus -1fil
\begin{IEEEbiography}[{\includegraphics[width=1in,height=1.25in,clip,keepaspectratio]{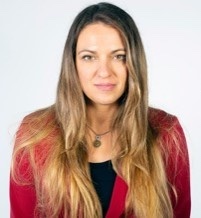}}]{Elisabetta Biondi}
Elisabetta Biondi is a researcher at the Institute for Informatics and Telematics (IIT) of the National Research Council of Italy (CNR). The general focus of her research activity is the mathematical modelling of human behaviours in social networks, with a special focus on mobility and information diffusion. Currently, she is working in the field of opinion diffusion. She is a member of the CNR research unit for the H2020  SoBigData++ and HumaneE-AI-Net projects and was involved in other H2020 and FP7 projects too.
\end{IEEEbiography}

\vskip -2.5\baselineskip plus -1fil

\begin{IEEEbiography}[{\includegraphics[width=1in,height=1.25in,clip,keepaspectratio]{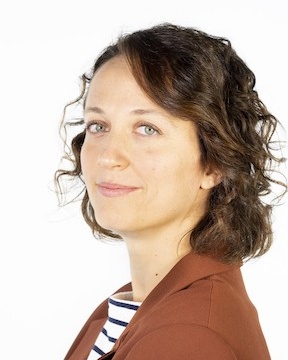}}]{Chiara Boldrini}
is a Senior Researcher at IIT-CNR. Her research interests are in decentralized AI, computational social sciences, mobile and ubiquitous systems. She has published 50+ papers on these topics. She is the IIT-CNR co-PI for H2020 SoBigData++ and HumaneE-AI-Net projects, and was involved in several EC projects since FP7. She is in the Editorial Board of Elsevier Pervasive and Mobile Computing and of Elsevier Computer Communications. She was the lead guest editor for the PMC Special Issue on IoT for Fighting COVID-19. She has served as TPC vice-chair of IEEE PerCom'21 and, over the years, has been in the organizing committee of  several IEEE and ACM conferences/workshops, including IEEE PerCom and ACM MobiHoc.
\end{IEEEbiography}

\vskip -2.5\baselineskip plus -1fil

\begin{IEEEbiography}[{\includegraphics[width=1in,height=1.25in,clip,keepaspectratio]{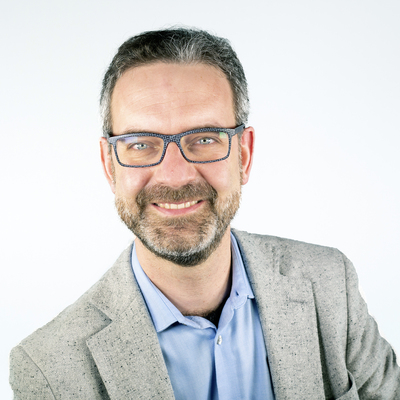}}]{Andrea Passarella} 
(PhD 2005) is a Research Director at the Institute for Informatics and Telematics (IIT) of CNR. Prior to joining IIT, he was with the Computer Laboratory of the University of Cambridge, UK. He has published 200+ papers on human-centric mobile networks, Online and Mobile social networks, opportunistic, ad hoc and sensor networks. He received four best paper awards, including at IFIP Networking 2011 and IEEE WoWMoM 2013. He was General Co-Chair of IEEE PerCom 2022 and WoWMoM 2019, and workshops co-chair for IEEE INFOCOM 2019. He was the PC co-chair of IEEE WoWMoM 2011, Workshops co-chair of several IEEE and ACM conferences. He is the founding Associate EiC of the Elsevier Journal Online Social Networks and Media (OSNEM). He is co-author of the book "Online Social Networks: Human Cognitive Constraints in Facebook and Twitter Personal Graphs" (Elsevier, 2015). He is the chair of the IFIP WG 6.3 "Performance of Communication Systems".
\end{IEEEbiography}

\vskip -2.5\baselineskip plus -1fil

\begin{IEEEbiography}[{\includegraphics[width=1in,height=1.25in,clip,keepaspectratio]{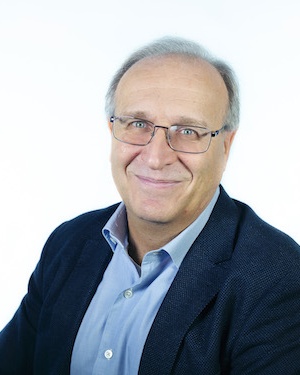}}]{Marco Conti}
is a CNR research director and, currently, he is the director of IIT-CNR institute. He has published more than 400 scientific articles related to design, modeling, and experimentation of computer and communication networks, pervasive systems, and online social networks. He is the founding EiC of Online Social Networks and Media, EiC for Special Issues of Pervasive and Mobile Computing, and, from 2009 to 2018, EiC of Computer Communications. He has received several awards, including the Best Paper Award at IFIP Networking 2011, IEEE ISCC 2012, and IEEE WoWMoM 2013. 
He was included in the "2017 Highly Cited Researchers" list compiled by Web of Science for the most cited articles in Computer Science. 
He served as the General/Program chair of several major conferences, including IFIP Networking 2002, IEEE WoWMoM 2005 and 2006, IEEE PerCom 2006 and 2010, ACM MobiHoc 2006, IEEE MASS 2007 and IEEE SmartComp 2021.
\end{IEEEbiography}

\end{document}